\documentclass{article}

\usepackage[preprint]{neurips_2026}

\usepackage[utf8]{inputenc} 
\usepackage[T1]{fontenc}    
\usepackage{hyperref}       
\usepackage{url}            
\usepackage{booktabs}       
\usepackage{amsfonts}       
\usepackage{nicefrac}       
\usepackage{microtype}      
\usepackage{xcolor}         

\usepackage{graphicx}  
\usepackage{amsmath}   
\usepackage{amsthm}
\usepackage{amssymb}   
\usepackage{hyperref}  
\usepackage{geometry}  
\usepackage{authblk}   
\usepackage{natbib}    
\usepackage{enumitem}  
\usepackage{amsfonts}
\usepackage{bbm}
\usepackage{mathtools}
\usepackage[vlined, ruled, linesnumbered]{algorithm2e}

\usepackage[suppress]{color-edits} 
\addauthor[sz]{sz}{blue}
\addauthor[claude]{claude}{red}

\newcommand{\hatv}{\hat{v}}

\newtheorem{lemma}{Lemma}

\newtheorem{theorem}{Theorem}

\newcommand{\eps}{\varepsilon}

\newcommand{\overv}{\overline{v}}
\newcommand{\underv}{\underline{v}}

\newcommand{\set}[1]{\left\{ #1 \right\}}

\newcommand{\attr}[2]{\mathsf{#1}.\mathsf{#2}}

\newcommand{\node}{\mathtt{node}}
\newtheorem*{theorem*}{Theorem}

\usepackage{physics}

\newtheorem{remark}{Remark}
\newtheorem{corollary}{Corollary}

\usepackage{cleveref}

\newcommand{\Reg}{\mathsf{Reg}}

\newcommand{\Rev}{\mathsf{Rev}}

\newcommand{\cM}{\mathcal{M}}

\newcommand{\cD}{\mathcal{D}}
\newcommand{\cI}{\mathcal{I}}

\title{The Power of Menus in Dynamic Pricing: Optimal Regret and Equivalence Results}

\author{%
  Shiliang Zuo\\
  Meta\\
  \texttt{szuo.rs@gmail.com} \\
}

\begin{document}

\maketitle

\begin{abstract}
  
We study dynamic pricing where a seller repeatedly interacts with a strategic, non-myopic buyer who has a fixed private valuation and discounts future utility. Prior work focused exclusively on posted-price mechanisms, where the seller gives a take-it-or-leave-it offer. For our first result, under geometric discounting $\gamma_t = \gamma^{t-1}$, we show that menu mechanisms consisting of allocation-payment contracts achieve $O(T_\gamma)$ regret, where $T_\gamma$ is the buyer's effective discounted time horizon. We also establish a $\Omega(T_\gamma)$ lower bound, demonstrating the bound is tight. Considering the geometric discounting buyer with a constant discount factor, our bound is $O(1)$, while prior bounds using posted-price mechanisms incur an unavoidable $\Omega(\log\log T)$ factor in regret. Our second contribution is more conceptual in nature. The problem of dynamic pricing sits at the intersection of two paradigms: learning with strategic agents in computer science / machine learning and revelation-principle-based mechanism design in economics, yet their relationship has remained unclear. We establish an equivalence between these two classes of mechanisms in our single-buyer setting: any adaptive indirect menu mechanism can be converted into an incentive-compatible direct mechanism with no less revenue, and conversely, any such direct mechanism can be implemented by an adaptive indirect menu mechanism with no less revenue. Consequently, the two classes achieve identical optimal regret. The adaptive, data-driven algorithms of online learning and explicit type elicitation from direct revelation type mechanisms are two languages towards solving the same problem.

\end{abstract}

\section{Introduction}

In this work we study the problem of \emph{dynamic pricing}: a seller repeatedly offers a product to a {single} buyer over $T$ rounds. The buyer has a private valuation $v$ for the item that remains fixed across time, while the seller seeks to maximize cumulative revenue. This problem was introduced by \cite{amin2013learning} and has since been studied extensively (\cite{kleinberg2003value, mohri2014optimal, drutsa2017horizon, haghtalab2022learning, drutsa2020optimal}). Despite this rich line of work, two significant gaps remain. First, all prior work focuses exclusively on posted-price mechanisms, yet there exists a much richer class of mechanisms based on menus---and it is unclear whether this richer class can yield improved regret. Second, there is a conceptual gap. Dynamic pricing has been studied from two seemingly difference perspectives. In machine learning, the problem is treated as strategic online learning and typically designs indirect adaptive learning-based mechanisms. In economics, the problem is treated as dynamic contracting and typically works with \emph{direct} revelation mechanisms. What is the precise relationship between these two classes of mechanisms? 
We elaborate on both gaps and present our contributions below.

\paragraph{Menu-based learning mechanisms.} 

\begin{table}[h] \centering \small \caption{Regret bounds for dynamic pricing ($\gamma_t = \gamma^{t-1}$, $\gamma < 1$). For lower bounds, \cite{amin2013learning} provide $\Omega(T_\gamma)$ in the non-myopic setting, while \cite{kleinberg2003value} provide $\Omega(\log \log T)$ in myopic setting. } \label{tab:comparison} \begin{tabular}{lcc} \toprule \textbf{Mechanism} & \textbf{Any $T, \gamma$} & \textbf{Constant $\gamma$ independent of $T$} \\ \midrule Posted-price \citep{drutsa2017horizon} & $O(T_\gamma \log T_\gamma \log\log T)$ & $O(\log\log T)$ \\ Posted-price \citep{haghtalab2022learning} & $O(\log T + T_\gamma \log T_\gamma)$ & $O(\log T)$ \\ \textbf{Menu, 2 items (Ours)} & $\boldsymbol{O(T_\gamma )}$ & $\boldsymbol{O(1)}$ \\ 
Lower bound for posted-price & \multicolumn{2}{c}{$\Omega(T_\gamma)$ / $\Omega(\log\log T)$} \\
\textbf{Lower bound for menu (Ours)} & \multicolumn{2}{c}{{$  \boldsymbol{\Omega(T_\gamma)}  $ }} \\
\bottomrule \end{tabular} \end{table}


Prior work has focused exclusively on \emph{posted-price mechanisms}, where the seller offers a single take-it-or-leave-it price in each round. Yet in principle, the seller can employ a much richer class of mechanisms: \emph{menus}. A menu mechanism offers the buyer a set of allocation-payment contracts, allowing the buyer to select among multiple options rather than making a binary accept/reject decision. {Besides its literal reading as a probability of allocation, a fractional allocation can represent the fraction of the item delivered when the item is divisible, or a service level---for instance, the higher and lower service tiers commonly offered on online marketplaces, where a lower tier provides reduced service at a lower price.} 

To illustrate the power of menus, consider first the simpler case of a myopic buyer. In the seminal work of \cite{kleinberg2003value}, it is shown that posted-price mechanisms can achieve regret $\Theta(\log \log T)$, and that such regret is necessary. However, with menus the seller can achieve regret $O(1)$. Specifically, the seller can design a menu $p: [0,1] \mapsto \mathbb{R}^+$ specifying a payment for each allocation probability. Once $p(\cdot)$ is specified, the buyer then chooses his preferred allocation probability. For example, setting $p(a) = a^2/2$ yields buyer utility $v \cdot a - p(a)$ in a single round, which is maximized at $a = v$. Thus, the buyer's choice immediately reveals their valuation. This raises our first research question:

\begin{center}
    \emph{Can menu mechanisms achieve better regret than posted-price mechanisms against non-myopic buyers?}
\end{center}

We answer this question affirmatively. For our regret upper bound, we focus on the geometric discounting setting where $\gamma_t = \gamma^{t-1}$ for $\gamma \in (0,1)$. We show that menu mechanisms achieve $O(T_\gamma)$ regret, where $T_\gamma$ denotes the buyer's effective time horizon. Hence, while it is believed since \cite{kleinberg2003value} that the dependence on $T$ could not be removed, we show this is not the case if the principal use menu-based mechanisms. In particular, our bound improves upon the best known bounds for posted-price mechanisms: $O(\log T + T_\gamma \log T_\gamma)$ (\cite{haghtalab2022learning}) and $O(T_\gamma \log T_\gamma \log \log T)$ (\cite{drutsa2017horizon}). For constant $\gamma < 1$, our regret bound becomes $O(1)$, compared to the $\Omega(\log \log T)$ lower bound for posted-price mechanisms. Further, our menu only requires two entries each round. The key insight is regret measures a \emph{fundamental tradeoff} between allocative efficiency for lower type agents and information rent paid to high type agents, and that using menus (even with just two entries) enable a fine-grained control of the two quantities.

\paragraph{Indirect Adaptive Mechanisms vs. Direct Revelation Mechanisms. } Our second contribution is more conceptual in nature. The dynamic pricing problem has been studied from two seemingly distinct perspectives.

The aforementioned machine learning / computer science perspective, treats the seller's problem as one of \emph{learning through interaction} \cite{kleinberg2003value, amin2013learning}. The seller posts prices sequentially, observes the buyer's accept/reject decisions, and adapts future prices based on this feedback. This approach appears algorithmic and data-driven, requiring no upfront communication from the buyer; instead, the buyer simply best-responds to the seller's pricing strategy. {In mechanism-design terms, such an algorithm is an adaptive \emph{indirect} mechanism: no explicit reports are solicited, and future outcomes are determined by the history of the buyer's choices.}

The economics perspective views this as a problem of \emph{dynamic contracting} \cite{bergemann2019dynamic, salanie2005economics, borgers2015introduction}. By the revelation principle, the seller can commit upfront to a \emph{direct revelation mechanism}: the buyer reports a type $\hat{v}$, and the seller commits to a complete sequence of allocations and payments $(a_t(\hat{v}), p_t(\hat{v}))_{t=1}^T$ contingent on this report. The mechanism must satisfy incentive compatibility (IC)---truthful reporting maximizes the buyer's utility---and individual rationality (IR). The seller's problem then reduces to optimizing revenue subject to these constraints.

{These two formulations look quite different on the surface. A direct mechanism relies on an explicit report of the buyer's type, while an adaptive indirect mechanism relies only on preferences revealed through the buyer's choices.}
The technical tools also differ substantially: mechanism design employs the revelation principle and solves structured optimization problems subject to incentive constraints, whereas online learning utilizes tools from exploration-exploitation tradeoffs and sequential decision-making under partial feedback.
{Yet both frameworks address the same underlying economic problem. This raises a natural question:}
\claudecomment{Meta-review W3 / Reviewer awqz: ``learning'' vs.\ ``revelation'' reframed as indirect vs.\ direct implementation in this model. } \szcomment{I dropped your edit, feel a bit terse. if you disagree add back and make an argument. }

\begin{center}
    \emph{What is the relationship between indirect adaptive mechanisms and direct revelation mechanisms for dynamic pricing against a non-myopic buyer?}
\end{center}

{We establish that the two classes are \emph{equivalent} in our setting. Specifically, we prove that any adaptive indirect menu mechanism can be converted into an IC and PIR direct revelation mechanism with the same revenue, and conversely, any IC and PIR direct mechanism can be converted into an adaptive indirect menu mechanism with no less revenue via what we call the \emph{price equalization procedure}. As a corollary, the optimal regret achievable by adaptive indirect mechanisms equals that achievable by direct mechanisms.}

{This equivalence is useful in both directions. The direct formulation characterizes the feasible outcomes through a continuum of IC and PIR constraints, which can be difficult to solve explicitly; an adaptive pricing algorithm instead provides a constructive sequential mechanism without solving this constraint system. Our first conversion shows that any such algorithm implicitly constructs a feasible direct mechanism. Our second conversion shows that, in this setting, explicit type reports provide no additional power over preferences revealed through menu choices: any IC and PIR direct mechanism can be implemented without soliciting reports. }

Utilizing the equivalence result, we provide a proof of a $\Omega(T_\gamma)$ lower bound on regret. Hence, our upper bound on regret is tight. This lower bound was first established by \cite{amin2013learning} for posted-price mechanisms. We show this bound also holds for menu-based mechanisms; also, our proof clearly demonstrates that at the heart of the problem is a tradeoff between two quantities: allocative inefficiency for lower types versus information rents for higher types.

\subsection{Related Work and Discussion}

Firstly, the specific problem we studied in this work is built on the line of work on dynamic pricing against non-myopic agents. The problem has received significant attention in the computer science community in recent years, categorized under the broader class of problems of \emph{learning with strategic agents}. The problem was originally introduced by \cite{amin2013learning}, with followup works including \cite{haghtalab2022learning, drutsa2017horizon, mohri2014optimal}. Extensions to multiple buyers include \cite{liu2018learning, drutsa2020reserve}. The different lens of exploring dynamic pricing against non-myopic agents is explored in {economics} literature falling under the domain of \emph{dynamic mechanism design} (\citet{bergemann2019dynamic}). For the case where the principal and agent share the same discount factor, there is a well-known result: the optimal dynamic mechanism is a static one that pose the optimal Bayesian price, which is known as \emph{false dynamics} \cite{borgers2015introduction}. \footnote{In fact, as \cite{borgers2015introduction} point out, if the seller and buyer have different discounting sequence with the buyer discounting future utility more heavily, then the gains from intertemporal trade can be infinite. Indeed, the seller and buyer should sign the following debt contract: the seller pays the buyer 1 dollar in time period 1, and ask the buyer to repay $1/\gamma^{T-1}$ dollars (minus an infinitesimal small amount) in time period $T$. If the buyer has geometric discounting with rate $\gamma$, then the contract is strictly beneficial to both parties. We avoid this scenario by allowing the agent to reserve the right to drop out at any round (see \Cref{sec:setup}). }

There also exists work that explore the problem from a pure learning perspective where the agent is myopic and best responds each round. The early work that introduced the problem include \cite{kleinberg2003value} (for which the literature on pricing against non-myopic agents was built upon), and later also extended to incorporate contextual information ~\cite{liu2021optimal, mao2018contextual}. Robust versions of the pricing algorithms were also explored in~\cite{krishnamurthy2021contextual, zuo2024corruption}.

Our work is also closely related the line of work on contract theory, in particular the use of menus in principal-agent models with adverse selection (see \cite{salanie2005economics, bolton2004contract, laffont1981theory}). Relatedly, some recent work in the computer science community studied computational aspects of menus in principal-agent problems (e.g. \cite{han2024learning}, \cite{guruganesh2023power}, \cite{castiglioni2022designing}).

Another related work is \cite{albert2025learning}; the focus of their work was on characterizing when the \emph{false dynamics} scenario can be escaped, i.e., what game structure allows the principal to reason about the agent's private type through repeated interactions (when they share the same discount factor). Another related work is \cite{conitzer2004computational}; their work was not on the dynamic setting where ``learning" is involved, and focus more on the computational perspective. They discuss the direct revelation mechanism and what they call an \emph{insincere} mechanism, and argue that in the direct revelation mechanism the computational burden is on the principal, while in the insincere mechanism the computational burden is shifted to the agent.

\section{Problem Setting: Dynamic Pricing against a Non-Myopic Agent}
\label{sec:setup}

There are two parties, a seller and a buyer. We may also refer to the seller as the principal and the buyer as the agent. They interact in discrete time for $T$ rounds. The buyer holds a private valuation 
\[
v\in [0, 1]
\]
for the item. At each round, the seller has one unit of the item she would like to sell to the buyer. The seller is patient and does not discount future utility. The buyer is impatient, modeled by a sequence of discount factors
$\{\gamma_t\}_{t=1}^T$ satisfying
\[
\gamma_1 \ge \gamma_2 \ge \cdots \ge \gamma_T > 0.
\]
We write
\[
T_\gamma := \sum_{t=1}^T \gamma_t
\]
for the buyer's effective discounted time horizon.
A canonical special case is geometric discounting,
\[
\gamma_t=\gamma^{t-1}, \qquad \gamma\in(0,1),
\]
for which
\[
T_\gamma=\frac{1-\gamma^T}{1-\gamma}.
\]
Our regret upper bound (\Cref{sec:menu}) is established for geometric discounting, whereas our equivalence results and our lower bound (\Cref{sec:menu-equiv}) hold for any discount sequence that is positive and non-increasing. We normalize $\gamma_1 = 1$; this is without loss of generality, since rescaling all discount factors by a common positive constant does not change the buyer's preferences. 

At each round the outcome is a pair of payment and allocation probabilities $(p_t, a_t)$. The seller's performance is measured by the strategic regret, as defined by the following:
\begin{align*}
\mathtt{Reg}(v) &= T\cdot v - \sum_{t=1}^T p_t, \quad \mathtt{Reg} = \sup_{v\in [0, 1]} \mathtt{Reg}(v)
\end{align*}
Here $\mathtt{Reg}(v)$ measures the difference of the optimal payment value the principal could have collected had she known the agent's true type, versus the actual payment collected.

The buyer's utility is given according to his discounting factor and the actual allocation and payment outcomes:
\[
\sum_{t=1}^T \gamma_t (a_t\cdot  v - p_t). 
\]
We now describe the classes of mechanisms that determine the actual payment and allocations $(p_t, a_t)$.

\subsection{Indirect Menu-based Learning Mechanisms} We begin by describing posted-price mechanisms, the focus of prior work, and then introduce the more general menu-based mechanisms studied in this paper. \paragraph{Posted-Price Mechanisms.} In a posted-price mechanism, the principal posts a price each round and the buyer makes a binary decision to accept or reject. Formally, the mechanism is described by a sequence of mappings $\mathcal{M} = (\mathcal{M}_t)_{t=1}^T$, where \[ \mathcal{M}_t: (p_\tau, a_\tau)_{\tau=1}^{t-1} \mapsto p_t, \] and $a_\tau \in \{0,1\}$ denotes the buyer's accept/reject decision at round $\tau$. The mechanism is announced to the buyer before the game begins; given $\mathcal{M}$, the buyer strategically responds to maximize his total utility over the $T$ rounds. Equivalently, such a mechanism can be represented as a \emph{binary decision tree} of depth $T$. Each non-leaf node is associated with a posted price. A rejection moves the buyer to the left child; an acceptance moves him to the right child. Leaf nodes encode complete interactions of length $T$. The buyer chooses the root-to-leaf path that maximizes his utility. 

\paragraph{Menu-Based Mechanisms.} A key departure in our work is allowing fractional allocation probabilities instead of binary allocation decisions. In a menu-based mechanism, at each round the principal posts a \emph{menu of contracts}: \[ \set{(a, p(a)): a\in A \subset [0,1]}. \] Here $a$ is an allocation probability and $p(a)$ is the associated payment. The agent selects an allocation level $a \in A$ and pays $p(a)$, upon which the item is allocated with probability $a$. The payment is made ex-ante regardless of the actual allocation outcome ex-post. Formally, a menu-based mechanism is described by mappings $\mathcal{M} = (\mathcal{M}_t)_{t=1}^T$, where \[ \mathcal{M}_t: (M_\tau, a_\tau)_{\tau=1}^{t-1} \mapsto M_t, \] with $M_t$ denoting the menu offered at round $t$ and $a_\tau \in [0,1]$ the allocation probability selected by the buyer at round $\tau$. Posted-price mechanisms are the special case where each menu contains only two options: $(1, p_t)$ (accept at price $p_t$) and $(0, 0)$ (reject).


\subsection{Direct Revelation Mechanism}
In a direct mechanism, the agent is first asked to report his private type  
(a ``bid''). The report $\hatv \in [0, 1]$ may differ from the  
agent’s true value $v$. For each report $\hatv$, the principal commits to a  
sequence of per-round contracts  
\[
  \hatv\mapsto (p_t(\hatv), a_t(\hatv))_{t=1}^T,
\]
which specify the payment $p_t(\hatv)$ and allocation probability $a_t(\hatv)$ in each  
period. We assume that payments are nonnegative, i.e., $p_t(\hat{v}) \ge 0$ for all reports $\hat{v}$ and time period $t$. 

Each pair $(p_t(\hatv), a_t(\hatv))$ functions as a spot contract: if the  
buyer continues participating in round $t$, he pays $p_t(\hatv)$ and receives  
allocation $a_t(\hatv)$. The seller and buyer executes this contract sequence  
\emph{sequentially over time}, but the agent always retains the right to  
\emph{drop out at any round}. Dropping out terminates the mechanism  
immediately. The interaction in the direct revelation mechanism proceeds as follows. 
\begin{enumerate}
\item The principal announces the direct mechanism to the agent. 
\item The agent makes a report $\hatv$ on his private type. 
\item At the beginning of each round $t$, the agent must decide whether to drop out or not. If he drops out the game immediately ends; otherwise he is charged $p_t(\hatv)$ and receives the item with probability $a_t(\hatv)$, then the game proceeds to the next round. 
\end{enumerate}

A direct mechanism must satisfy two key conditions:

\begin{itemize}
  \item \emph{Incentive compatibility (IC):} truthful reporting maximizes the
    agent’s discounted utility among all possible reports.
  \item \emph{Periodic individual rationality (PIR):} at every round, a truthful
    agent weakly prefers continuing participation over dropping out. 
\end{itemize}

Together, IC and PIR ensure that an agent is never worse off by reporting his true type and remaining in the mechanism throughout. Similar PIR constraints have also appeared in prior work, e.g. \cite{bergemann2010dynamic, balseiro2018dynamic}. 

To formalize these requirements, we define the discounted utility of a type-$v$
agent who reports $\hatv$ and accepts all contracts from periods $1$ through
$t$ as
\[
  U_{1:t}(\hatv; v)
    := \sum_{\tau = 1}^{t}
       \gamma_\tau [a_\tau(\hatv) v - p_\tau(\hatv) ].
\]
IC and PIR together then require that for all types $v$, all reports $\hatv$, and all
periods $t = 0, 1, \dots, T$, we have (defining $U_{1:0}(\cdot; \cdot)$ = 0): 
\begin{equation}
  U_{1:T}(v; v) \;\ge\; U_{1:t}(\hatv; v).
\end{equation}

\paragraph{Tie-breaking.} {In an indirect mechanism, the buyer may have multiple strategies that maximize his total discounted utility, and these may generate different revenue for the seller. Throughout the paper we adopt the convention that the buyer breaks ties in favor of the seller: among all strategies maximizing his total discounted utility over the complete interaction, the buyer selects one that maximizes the seller's total revenue. This is the dynamic analogue of the principal-favorable tie-breaking rule commonly used in Stackelberg equilibria. The direct formulation uses the analogous convention: IC and PIR only require truthful reporting and continued participation to be \emph{weakly} optimal, and the buyer reports truthfully and follows the prescribed contract sequence. Under this convention, the outcome generated by each type $v$, and hence the seller's revenue, is well defined for both classes of mechanisms.}

\section{Optimal Regret via Menus}
\label{sec:menu}
We now demonstrate that by employing menus of contracts, the seller
can achieve a regret bound of $O(T_\gamma)$, hence breaking the
$T$-dependent $\Omega(\log\log T)$ barrier present in prior work
using posted-price mechanisms. We later establish an
$\Omega(T_\gamma)$ lower bound, showing that this rate is tight up
to constant factors.

For the upper bound in this section, we specialize to geometric
discounting,
\[
\gamma_t=\gamma^{t-1}, \qquad \gamma\in(0,1). \footnote{
The phase-consistency argument can be extended beyond geometric
discounting when the discount sequence has sufficiently controlled
tails. In particular, for a phase beginning at round $t$ with
candidate-interval width $\delta$, it suffices to choose a phase
length $D$ such that
\[
\sum_{s=t+D}^{T}\gamma_s
\le
c\,\delta^2
\sum_{s=t}^{\min\{t+D-1,T\}}\gamma_s
\]
for a sufficiently small absolute constant $c>0$.
The resulting regret depends on the phase lengths required to satisfy
this tail-decay condition. We do not pursue a general regret bound
for arbitrary discount sequences here.
}
\]



\subsection{Menu-based Learning Mechanism Design}
\label{sec:two-item-menu}

\begin{algorithm}[t]
\caption{Consistent Two-Item Menu Search}
\KwIn{Initial candidate interval $[0,1]$, discount parameter $\gamma$,
constant $\rho>0$}

Initialize $\underline v_1 \gets 0$ and $\overline v_1 \gets 1$\;
Set $t\gets 1$\;

\While{$t\le T$}{
    $\delta_t \gets \overline v_t-\underline v_t$\;
    $m_t \gets (\underline v_t+\overline v_t)/2$\;

    Define the two menu entries
    \[
    L_t \gets
    \bigl(1-\delta_t,(1-\delta_t)\underline v_t\bigr),
    \qquad
    H_t \gets
    \bigl(1,\underline v_t+\delta_t^2/2\bigr).
    \]

    Set
    \[
    D_t \gets
    \left\lceil
    \frac{
    \log\!\left(
    \frac{1+\rho\delta_t^2}{\rho\delta_t^2}
    \right)
    }{
    \log(1/\gamma)
    }
    \right\rceil,
    \qquad
    \overline D_t \gets \min\{D_t,T-t+1\}.
    \]

    Post the two-item menu $M_t=\{L_t,H_t\}$ in the first round
    of the phase\;

    Let $C_t\in\{L_t,H_t\}$ be the agent's choice in the first
    round of this phase\;

    \For{$s=t+1,\ldots,t+\overline D_t-1$}{
        Post the same menu $M_t=\{L_t,H_t\}$\;

        \If{the agent's choice is not $C_t$}{
            Offer only the outside option $(0,0)$ in every future round\;
        }
    }

    \If{$C_t=L_t$}{
        Update
        \[
        [\underline v_{t+D_t},\overline v_{t+D_t}]
        \gets
        [\underline v_t,m_t+\delta_t/8]
        \cap
        [\underline v_t,\overline v_t].
        \]
    }
    \Else{
        Update
        \[
        [\underline v_{t+D_t},\overline v_{t+D_t}]
        \gets
        [m_t-\delta_t/8,\overline v_t]
        \cap
        [\underline v_t,\overline v_t].
        \]
    }

    $t\gets t+D_t$\;
}
\label{alg:menu-two-item}
\end{algorithm}

We now present a two-item menu algorithmic learning mechanism that achieves the optimal regret
bound. The mechanism is summarized in \Cref{alg:menu-two-item}. At a high level, the seller maintains an interval
$I_t=[\underv_t,\overv_t]$ that contains the buyer's value and repeatedly uses
a two-item menu to ask whether the value lies above or below the midpoint of
the interval.

Let
\[
    \delta_t := \overv_t-\underv_t,
    \qquad
    m_t := \frac{\underv_t+\overv_t}{2}.
\]
The menu offered in the phase is
\[
    L_t =
    \bigl(1-\delta_t,(1-\delta_t)\underv_t\bigr),
    \qquad
    H_t =
    \bigl(1,\underv_t+\delta_t^2/2\bigr).
\]
For a buyer of value $v$, the utility difference between the two options is
\[
\begin{aligned}
    U_{H_t}(v)-U_{L_t}(v)
    &=
    \left(v-\underv_t-\frac{\delta_t^2}{2}\right)
    -(1-\delta_t)(v-\underv_t)  \\
    &=
    \delta_t(v-\underv_t)-\frac{\delta_t^2}{2} \\
    &=
    \delta_t(v-m_t).
\end{aligned}
\]
Thus a myopic buyer with $v>m_t$ prefers $H_t$, while a myopic buyer with
$v<m_t$ prefers $L_t$.

The mechanism proceeds in phases, and repeats this same two-item menu for an entire phase.  The buyer's
first choice determines the branch of the binary search.  If the buyer ever
switches between $L_t$ and $H_t$ within the phase, the seller offers only the
outside option $(0,0)$ in every future round.  This consistency requirement
amplifies the cost of inducing the wrong branch: a buyer who wants to be
misclassified must choose the wrong menu item throughout the phase, rather than
only once.

Fix a sufficiently small absolute constant $\rho>0$.  For a phase
beginning at round $t$ with interval width $\delta_t$, define the
\emph{planned} phase length
\[
D_t :=
\left\lceil
\frac{
\log\!\left(
\frac{1+\rho\delta_t^2}{\rho\delta_t^2}
\right)
}{
\log(1/\gamma)
}
\right\rceil
\]
and the \emph{actual} phase length
\[
\overline D_t := \min\{D_t,T-t+1\}.
\]
The seller holds the menu $\{L_t,H_t\}$ fixed for
$\overline D_t$ rounds.  If $\overline D_t<D_t$, then the phase
reaches the end of the horizon and the mechanism terminates without
performing another interval update.  Otherwise the full planned
phase is completed, and the interval is updated according to the
buyer's consistent choice:
\[
I_{t+D_t}
=
[\underline v_t,m_t+\delta_t/8]
\]
after choosing $L_t$, and
\[
I_{t+D_t}
=
[m_t-\delta_t/8,\overline v_t]
\]
after choosing $H_t$.
Both successor intervals have width at most
\[
\frac58\delta_t.
\]

Although the planned length $D_t$ need not be bounded by the finite
horizon, the actual phase length satisfies
\[
\overline D_t
=
O\!\left(
T_\gamma\left(1+\log\frac1{\delta_t}\right)
\right).
\]


\subsection{Analysis}

We state the key results for the main regret bound, the full proof is deferred to Appendix~\ref{sec:proof-menu-tight}. We first show that the candidate interval $[\underv_t, \overv_t]$ always contains the true hidden value $v$. 
\begin{lemma}[Consistency of the phase]
\label{lem:phase-consistency}
There exists an absolute constant $\rho>0$ such that the following holds in
every phase.  If $v\ge m_t+\delta_t/8$, then the buyer weakly prefers choosing
$H_t$ consistently throughout the phase to choosing $L_t$ consistently and
being placed in the left branch.  Similarly, if $v\le m_t-\delta_t/8$, then the
buyer weakly prefers choosing $L_t$ consistently throughout the phase to
choosing $H_t$ consistently and being placed in the right branch.
\end{lemma}

The overlap between the two successor intervals handles types close to the
midpoint.  If
\[
    v\in[m_t-\delta_t/8,m_t+\delta_t/8],
\]
then either successor interval still contains $v$, so either branch is
acceptable for the purpose of future learning.

Next, the regret of each phase can bounded by the phase's duration and the length of the candidate interval during the phase. 
\begin{lemma}[Regret of one phase]
\label{lem:phase-regret}
The regret incurred during a phase of width $\delta_t$ and actual
duration $\overline D_t$ is
\(
O(\delta_t\overline D_t).
\)
\end{lemma}

\begin{theorem}[Optimal regret of two-item menus]
\label{thm:menu-regret-tight}
The two-item menu mechanism described above achieves regret
\(
    O(T_\gamma).
\)
\end{theorem}

\section{Equivalence Results}
\label{sec:menu-equiv}
Here we show that \emph{indirect menu-based learning} mechanisms are equivalent to {IC and PIR} {direct revelation} mechanisms in the single-buyer dynamic pricing model. 
Given a mechanism $\cM$ (either a direct revelation type or an indirect learning type), we define the revenue to the seller as the total collected monetary payment when an agent with private type (value) $v$ best responds to the mechanism:
$\Rev(v; \cM) = \sum_{t=1}^T p_t$. We show that for any indirect mechanism, there exists a direct mechanism whose revenue is at least no less than the indirect one for any type; and vice versa. Results in this section does not require geometric discounting and hold for any nonincreasing and positive discount sequence. 


\subsection{Converting Indirect Menu Mechanism to Direct Mechanism}

The menu-based indirect learning mechanism can be converted to a direct revelation type mechanism by the revelation principle. For each type $v \in [0, 1]$, let 
\[
H(v):=(a_t(v), p_t(v))_{t=1:T}
\]
be the trajectory generated by the principal's menu mechanism and the agent's best response strategy in the indirect mechanism.

\begin{theorem}\label{thm:menu-indirect-to-direct}
For any indirect menu-based learning mechanism $\cI$, there exists a direct mechanism $\cD$ with the same outcome and revenue, i.e, $\Rev(v, \cD) = \Rev(v, \cI), \forall v$. In particular the trajectory $H(v)$ from the indirect mechanism is itself an IC and PIR direct mechanism.
\end{theorem}
The proof is by an application of the revelation principle, which we defer to Appendix \ref{app:proof-menu-convert}. An interesting implication is that any indirect adaptive learning-based mechanism {implicitly constructs a feasible solution to the IC and PIR constraint system of the direct formulation}.

\subsection{Converting Direct Mechanism to Indirect Menu Mechanism}
\begin{algorithm}[H]

    \caption{Converting direct mechanism to menu-based learning mechanism}\label{alg:menu-direct-to-indirect}

    Input: direct mechanisms $v\mapsto (p_t(v), a_t(v))_{t=1}^{T}$\;
    
    \tcp{Construct a tree. Each node has two attributes: $\mathsf{values}$ and $\mathsf{menu}$}
    Construct root node $\mathsf{root}$ (depth is 1), set $\attr{\mathsf{root}}{values} = [0, 1]$\;
    
    \tcp{Price Equalization for Menu Construction}
    \For{$t \gets 1$ \KwTo $T$} {
    \label{algo:menu-adjustment-phase}
        \For {$\node$ in depth $t$} {
            \tcc{Iterate over all nodes at level $t$, compute the correct menu for the node}
            $S = \attr{node}{values}$\;
            Let $A$ be the set of allocation probabilities that occur in the set $S$, i.e., $A = \{a_t(v) : v \in S\}$\;
            \For {$a \in A$} {
                Let $S(a) = \{v \in S : a_t(v) = a\}$ be the set of reports whose allocation probability is $a$\;
                Let {$p_{\inf}(a) = \inf_{v \in S(a)} p_t(v)$ be the infimum of the prices} that occur in the direct mechanism\;
                Set $p_{\inf}(0) = 0$\;
                \tcc{Adjustment Steps: equalize prices and defer differences}
                \label{algoline:adjustment-steps}
                \For {$v \in S(a)$} {
                    \If{$t < T$}{
                    \label{alg-line:menu-adjust} Set $p_{t+1}(v) := p_{t+1}(v) + \frac{\gamma_t}{\gamma_{t+1}} \cdot (p_t(v) - p_{\inf}(a))$\;
                    }
                    Set $p_t(v) := p_{\inf}(a)$\;
                }
            }
            Set the menu for the current node as $\attr{node}{menu} = \{(a, p_{\inf}(a)) : a \in A\}$\;
            \For{$a \in A$}{
                Create child node $\mathsf{child}_a$ with $\attr{child_a}{values} = S(a)$\;
            }
        }
    }
\end{algorithm}

We now establish the converse direction: any direct mechanism can be implemented as an indirect learning mechanism.

\begin{theorem}\label{thm:menu-direct-to-indirect}
    Given any direct mechanism $\cD: v \mapsto (a_t(v), p_t(v))_{t=1:T}$ that is IC and PIR,  Algorithm~\ref{alg:menu-direct-to-indirect} produces an indirect menu-based mechanism $\cI$ satisfying: $ \forall v,\, \mathrm{Rev}(v, \mathcal{I}) \geq \mathrm{Rev}(v, \mathcal{D})$.
\end{theorem}


The conversion from direct to indirect is more involved than the reverse direction (which follows directly from the revelation principle). The core difficulty is the following: in a direct mechanism, the seller knows the buyer's reported type, so she can charge \emph{different payments to different types even when they receive the same allocation}. In an indirect mechanism, the seller never learns the buyer's type explicitly---she can only post a menu of (allocation, payment) pairs, and every buyer facing the same menu must pay the same price for the same allocation. We introduce a procedure (\Cref{alg:menu-direct-to-indirect}) that converts any direct mechanism into an indirect one. Note that in the procedure, for direct mechanisms, payments $p_t(\hat{v})$ are indexed by the reported type, whereas in the indirect mechanism, payments $p(a)$ depend only on the chosen allocation---this asymmetry is precisely what makes the conversion non-trivial.

The procedure essentially shows how to ``unfold'' any direct mechanism into a \emph{tree-structured} indirect mechanism. Each node in the tree corresponds to a particular interaction history, and branches correspond to the buyer's allocation choices. At each node we track: \begin{enumerate} \item \textbf{Candidate set} $\attr{node}{values}$: this is the set of types which would have reached this node in the interaction history. At each depth, the sets are disjoint and the union is the full set $[0,1]$. \item \textbf{Menu} $\attr{node}{menu}$: the (allocation, payment) pairs offered at this history. \end{enumerate} The central challenge is \textit{price equalization}. Consider a node whose candidate set contains types $v$ and $v'$ that receive the same allocation $a_t(v) = a_t(v') = a$ but pay different prices $p_t(v) \neq p_t(v')$. Since the indirect mechanism must post a single price for allocation $a$, we cannot directly replicate both contracts. Our solution is to set the menu price to the \emph{lowest} price in each allocation group, \[ p_{\inf}(a) = \mathop{\inf}_{v \in S(a)} p_t(v), 
 \quad \text{where } S(a) = \{v \in \attr{node}{values} : a_t(v) = a\}, \]  and \emph{defer} the remaining payment to the next round. Specifically, for any type $v$ that originally paid more than $p_{\inf}(a_t(v))$, we add the difference---adjusted for discounting by the factor $\gamma_t / \gamma_{t+1}$---to that type's payment in round $t+1$. This deferral is feasible because IC guarantees that each type has sufficient continuation utility to absorb the extra charge. After each adjustment step, the buyer's total discounted utility is unchanged, while the seller's revenue weakly increases (since payments are moved earlier in effective terms). The complete proof appears in Appendix~\ref{app:proof-menu-convert}.

The conversions (\Cref{thm:menu-indirect-to-direct}, \Cref{thm:menu-direct-to-indirect}) between indirect learning mechanisms and direct mechanisms indeed implies that the two classes of mechanisms are equivalent in power. 
\begin{corollary}
Given any direct mechanisms $\cD$ with regret $\Reg(\cD)$, there exists an indirect menu-based learning mechanism with regret $\Reg(\cI)$ no larger than that of the direct one; and vice versa. 
\end{corollary}

Note that while the conversion in \Cref{thm:menu-direct-to-indirect} can strictly \emph{increase} revenue, it does not mean that indirect mechanisms are more powerful; it is simply because the input direct mechanism was suboptimal. The revenue gain comes from deferring early payments to later rounds, which increases revenue while keeping the 
buyer indifferent (his discounted utility is unchanged). An optimally designed direct 
mechanism would already exploit this, leaving no slack for the conversion to improve 
upon. Thus neither class dominates the other, and the optimal regret is identical.

If we focus on the class of (indirect learning) posted-price mechanisms, one can show (via the same procedure) that they are equivalent to direct revelation mechanisms \emph{with binary allocation outcomes}. 
\begin{corollary}
    The conversions (\Cref{thm:menu-indirect-to-direct}, \Cref{thm:menu-direct-to-indirect}) holds with ``indirect menu-based learning mechanisms" replaced by ``posted-price mechanisms", and ``direct mechanisms" by ``direct mechanisms with binary allocation outcomes". 
\end{corollary}

\subsection{Lower Bound on Regret}
\label{subsec:lower-bound}
In \citet{amin2013learning}, they prove a lower bound of $\Omega(T_\gamma)$ for non-myopic agents in the posted-price mechanism setting. We show this lower bound still holds for menu-based mechanisms. Hence, our $O(T_\gamma)$ upper bound is tight up to constant factors. Our proof clearly pinpoints how the regret is measuring a \emph{fundamental trade-off} between allocation inefficiencies for lower-type agents and information rents paid to higher-type agents. We defer the full analysis to Appendix~\ref{sec:proof-lower-bound}. 

\begin{theorem}
{For any discount sequence $1 = \gamma_1 \ge \gamma_2 \ge \cdots \ge \gamma_T > 0$, the} regret for indirect menu-based learning mechanisms (and equivalently, direct revelation mechanisms) can be lower bounded by $\Omega(T_\gamma)$.
\end{theorem}

\section{Conclusion}
\label{sec:discussion}
In this work, we have shown two main results. For our first result, we show that menu-based mechanisms break the $\Omega(\log \log T)$ barrier for the dynamic pricing problem, achieving regret $O(T_\gamma)$ {under geometric discounting}, matching the lower bound up to constant factors.
Our second result is {an equivalence between indirect menu mechanisms and direct revelation mechanisms in the single-buyer dynamic pricing setting.} We mention two directions for future research. The first is to characterizing when direct mechanisms can be implemented indirectly. While we showed that any IC and PIR direct mechanism can be converted to an indirect mechanism in our setting, the general question remains: {what structural properties of a mechanism design problem enable this conversion?} The second would be to extend our results to multiple buyers and richer type spaces. In particular, does our equivalence results still hold for multiple buyers or multi-dimension type spaces? 



\bibliographystyle{plainnat}
\bibliography{references}  


\appendix
\section{Proofs for \Cref{sec:menu}}
\label{sec:proof-menu-tight}

\begin{proof}[Proof of \Cref{lem:phase-consistency}]

We first show that an optimal buyer never switches between
$L_t$ and $H_t$ within a phase.
Suppose the true value $v$ lies in the current candidate
interval $I_t$, and let $C^* \in \{L_t,H_t\}$ be the
myopically preferred contract (breaking ties in favor of
$H_t$).

Since the menu remains fixed throughout the phase,
consistently choosing $C^*$ weakly dominates any switching
strategy in terms of utility earned during the phase.
Moreover, the successor interval corresponding to $C^*$
contains $v$. In every subsequent phase, the buyer can
secure nonnegative utility by consistently choosing a
myopically preferred contract, since
\[
U_{L_t}(v) = (1-\delta_t)(v-\underline{v}_t) \ge 0
\]
whenever $v \in I_t$.
In contrast, switching triggers termination and yields
zero continuation utility.

We now prove the lemma for $v \ge m_t + \delta_t / 8$; the other case is symmetric. Suppose
$v\ge m_t+\delta_t/8$. Then choosing $L_t$ instead of $H_t$ in any round of
the phase causes current utility loss
\[
    U_{H_t}(v)-U_{L_t}(v)
    =
    \delta_t(v-m_t)
    \ge
    \frac{\delta_t^2}{8}.
\]
Let
\[
    \Gamma_{\mathrm{phase}}
    :=
    \sum_{i=0}^{D_t-1}\gamma^i
    =
    \frac{1-\gamma^{D_t}}{1-\gamma}
\]
and
\[
    \Gamma_{\mathrm{tail}}
    :=
    \sum_{i=D_t}^{\infty}\gamma^i
    =
    \frac{\gamma^{D_t}}{1-\gamma}.
\]
If the buyer falsely chooses $L_t$ consistently throughout the phase, his
discounted utility loss during the phase is at least
\[
    \frac{\delta_t^2}{8}\Gamma_{\mathrm{phase}}.
\]

We next upper bound the possible continuation gain from being placed in the
wrong branch.  Since values, allocations, and payments are bounded, the
buyer's per-round continuation utility can change by at most an absolute
constant $C$.  Hence the total discounted continuation gain from forcing the
wrong branch is at most
\[
    C\Gamma_{\mathrm{tail}}.
\]

By the definition of $D_t$,
\[
    \gamma^{D_t}
    \le
    \frac{\rho\delta_t^2}{1+\rho\delta_t^2}.
\]
Therefore
\[
\begin{aligned}
    \Gamma_{\mathrm{tail}}
    =
    \frac{\gamma^{D_t}}{1-\gamma}
    &\le
    \rho\delta_t^2
    \frac{1-\gamma^{D_t}}{1-\gamma}  \\
    &=
    \rho\delta_t^2\Gamma_{\mathrm{phase}}.
\end{aligned}
\]
Thus the continuation gain from false branching is at most
\[
    C\rho\delta_t^2\Gamma_{\mathrm{phase}}.
\]
Choosing $\rho\le 1/(8C)$ makes this no larger than the current utility loss
from choosing $L_t$ rather than $H_t$ throughout the phase.  Hence false
branching is not profitable.
\end{proof}

\begin{proof} [Proof of \Cref{lem:phase-regret}]
Fix a phase with interval $I_t=[\underv_t,\overv_t]$ and width $\delta_t$.
Since $v\in[\underv_t,\overv_t]$, the full-information benchmark revenue in
any round is at most
\[
    v\le \underv_t+\delta_t.
\]
The low menu option generates revenue
\[
    (1-\delta_t)\underv_t
    \ge
    \underv_t-\delta_t,
\]
and the high menu option generates revenue
\[
    \underv_t+\delta_t^2/2
    \ge
    \underv_t.
\]
Thus every nonzero menu choice generates revenue at least
$\underv_t-\delta_t$, and the seller's per-round regret is $O(\delta_t)$.
Since the actual phase duration is $\overline D_t$, the total regret
incurred in the phase is
\[
O(\delta_t\overline D_t).
\]
\end{proof}

\begin{proof}[Proof of \Cref{thm:menu-regret-tight}]
Index the phases by $k=0,1,2,\ldots$, and let $\delta_k$ denote
the candidate-interval width at the beginning of phase $k$.
Every completed phase shrinks the interval width by a factor at most
\[
\beta:=\frac58.
\]
There is at most one truncated phase, namely the final phase.
Consequently, for every phase that is actually reached,
\[
\delta_k\le \beta^k.
\]

We first bound the actual duration of a phase.  Write
\[
A_k
:=
\log\!\left(
\frac{1+\rho\delta_k^2}{\rho\delta_k^2}
\right)
=
\log\!\left(
1+\frac{1}{\rho\delta_k^2}
\right).
\]
Since $\rho>0$ is an absolute constant,
\[
A_k
=
O\!\left(1+\log\frac1{\delta_k}\right).
\]
By definition,
\[
D_k
\le
1+\frac{A_k}{\log(1/\gamma)},
\qquad
\overline D_k
=
\min\{D_k,T-t_k+1\},
\]
where $t_k$ is the first round of phase $k$.  Hence
\[
\overline D_k
=
O\!\left(
\left(1+\log\frac1{\delta_k}\right)
\min\left\{
T,\frac{1}{\log(1/\gamma)}
\right\}
\right).
\]
Using
\[
\log(1/\gamma)\ge 1-\gamma,
\]
we obtain
\[
\min\left\{
T,\frac{1}{\log(1/\gamma)}
\right\}
\le
\min\left\{
T,\frac{1}{1-\gamma}
\right\}.
\]

We next relate the latter quantity to $T_\gamma$.  Let
$x:=1-\gamma$.  Since
\[
\gamma^T=(1-x)^T\le e^{-xT},
\]
we have
\[
T_\gamma
=
\frac{1-\gamma^T}{1-\gamma}
\ge
\frac{1-e^{-xT}}{x}.
\]
For every $y\ge0$,
\[
1-e^{-y}
\ge
(1-e^{-1})\min\{y,1\}.
\]
Applying this with $y=xT$ gives
\[
T_\gamma
\ge
(1-e^{-1})
\min\left\{
T,\frac1{1-\gamma}
\right\}.
\]
Therefore
\[
\overline D_k
=
O\!\left(
T_\gamma
\left(1+\log\frac1{\delta_k}\right)
\right).
\]

By Lemma~2, the regret in phase $k$ is thus
\[
O\!\left(
T_\gamma\,
\delta_k
\left(1+\log\frac1{\delta_k}\right)
\right).
\]
The function
\[
x\mapsto x\left(1+\log\frac1x\right)
\]
is increasing on $(0,1]$.  Since $\delta_k\le\beta^k$,
\[
\delta_k
\left(1+\log\frac1{\delta_k}\right)
\le
\beta^k
\left(
1+k\log\frac1\beta
\right).
\]
Summing over all phases, including the possibly truncated final
phase,
\[
\begin{aligned}
\Reg
&\le
O(T_\gamma)
\sum_{k\ge0}
\beta^k
\left(
1+k\log\frac1\beta
\right) \\
&=
O(T_\gamma),
\end{aligned}
\]
because both
$\sum_{k\ge0}\beta^k$ and
$\sum_{k\ge0}k\beta^k$ are finite.
Thus the mechanism achieves regret $O(T_\gamma)$ for every finite
horizon $T$.
\end{proof}
\section{Proof of \Cref{sec:menu-equiv}}
\label{app:proof-menu-convert}

\begin{proof} [Proof of \Cref{thm:menu-indirect-to-direct}]
Fix any true type $v$, any potential misreport $\hat{v}$, and any round $t \in \{1, \ldots, T\}$. We must show:
\[
U_{1:T}(v; v) \geq U_{1:t}(\hat{v}; v).
\]

Consider the indirect menu mechanism. At each round $\tau$ and history $h_\tau$, the principal posts a menu $M_\tau(h_\tau) = \{(a, p_\tau(a; h_\tau)) : a \in [0,1]\}$. 

By construction, $(a_\tau(v), p_\tau(v))_{\tau=1:T}$ is the trajectory generated when a type-$v$ agent plays their best response in the indirect mechanism. This means at each round $\tau$, given the history $h_\tau(v)$ generated by type $v$'s play through round $\tau-1$, the agent selects allocation $a_\tau(v)$ from menu $M_\tau(h_\tau(v))$ and pays $p_\tau(v) = p_\tau(a_\tau(v); h_\tau(v))$. 

Now consider an alternative strategy in the indirect mechanism where the type-$v$ agent:
\begin{enumerate}
\item Mimics type $\hat{v}$'s behavior for rounds $1$ through $t$: at each round $\tau \leq t$, follows history $h_\tau(\hat{v})$ and selects allocation $a_\tau(\hat{v})$ from menu $M_\tau(h_\tau(\hat{v}))$, paying $p_\tau(\hat{v})$.
\item Drops out at round $t+1$ by selecting the outside option $(0,0)$ in all subsequent rounds.
\end{enumerate}

This deviation strategy yields utility exactly $U_{1:t}(\hat{v}; v)$ to the type-$v$ agent.

Since $(a_\tau(v), p_\tau(v))_{\tau=1:T}$ is generated by the agent's {best response} in the indirect mechanism, it must weakly dominate all alternative strategies, including the deviation described above:
\[
U_{1:T}(v; v) \geq U_{1:t}(\hat{v}; v).
\]

This holds for all $v$, all $\hat{v}$, and all $t \in \{1, \ldots, T\}$, establishing that the direct mechanism satisfies IC and PIR. {The argument applies to any selection of best responses $H(v)$, so this direction does not depend on how the buyer breaks ties.}
\end{proof}

\begin{proof} [Proof of \Cref{thm:menu-direct-to-indirect}]
The bulk of the proof is showing the mechanism still satisfies IC and PIR {after} each adjustment step (Parts (i) and (ii)) and that no adjustment is needed in the final period (Part (iii)). Then, in the final indirect mechanism, the prescribed path is a best response for each type, and by the tie-breaking convention the buyer's selected best response yields at least the revenue of this path (Part (iv)). 

Consider a single adjustment at round $t < T$ where we set $p_t(v) \leftarrow p_{\min}(a_t(v))$ and $p_{t+1}(v) \leftarrow p_{t+1}(v) + \frac{\gamma_t}{\gamma_{t+1}}(p_t^{\text{old}}(v) - p_{\min}(a_t(v)))$, where $p_t^{\text{old}}(v)$ denotes the value before adjustment.

Recall that the utility for a type-$v'$ agent reporting $\hat{v}$ and participating through round $\tau$ is:
\[
U_{1:\tau}(\hat{v}; v') = \sum_{s=1}^{\tau} \left[\gamma_s a_s(\hat{v}) \cdot v' - \gamma_s p_s(\hat{v})\right].
\]

We must verify that the adjusted mechanism still satisfies IC and PIR: for all types $v'$, all reports $\hat{v}$, and all rounds $\tau$,
\[
U_{1:T}(v'; v') \geq U_{1:\tau}(\hat{v}; v').
\]

\textbf{Part (i). Utility for truthful reporting for type $v$ is unchanged. } For a type-$v$ agent reporting $v$ and participating through all $T$ rounds, the utility change from the adjustment is:
\begin{align*}
\Delta U_{1:T}(v; v) &= \left[-\gamma_t p_{\min}(a_t(v)) - \gamma_{t+1}\left(p_{t+1}(v) + \frac{\gamma_t}{\gamma_{t+1}}(p_t^{\text{old}}(v) - p_{\min}(a_t(v)))\right)\right] \\
&\quad - \left[-\gamma_t p_t^{\text{old}}(v) - \gamma_{t+1} p_{t+1}(v)\right] \\
&= -\gamma_t p_{\min}(a_t(v)) - \gamma_t(p_t^{\text{old}}(v) - p_{\min}(a_t(v))) + \gamma_t p_t^{\text{old}}(v) \\
&= 0.
\end{align*}

Therefore, $U_{1:T}(v; v)$ remains unchanged. The allocation benefits $\sum_{s=1}^T \gamma_s a_s(v) \cdot v$ are unaffected, and the payment stream is adjusted to preserve its present value.

\textbf{Part (ii). Deviations remain suboptimal. }

Next, We need to verify that for any type $v'$ and any dropout round $\tau$, the constraint $U_{1:T}(v'; v') \geq U_{1:\tau}(v; v')$ continues to hold after adjusting $v$. 



We need to check how $U_{1:\tau}(v; v')$ changes. This depends on whether $\tau < t$, $\tau = t$, or $\tau \ge t+1$. 

\begin{itemize}

\item If $\tau < t$. Since we are adjusting price in round $t$ and the next round, dropping out before round $t$ have the exact same utility as before the adjustment. If the original mechanism were IC and PIR, dropping out before round $t$ would not benefit the agent. 



\item If $\tau = t$. We must show that a type-$v'$ agent cannot benefit from reporting $v$ and dropping out at round $t$ after the price adjustment. {Since $S(a)$ may be infinite, the infimum defining $p_{\min}(a)$ in \Cref{alg:menu-direct-to-indirect} may not be attained; to emphasize this, in this step we write} \[ p_{\inf}(a) := \inf_{w \in S(a)} p_t(w). \] Fix any $\epsilon > 0$. By definition of infimum, there exists $v^\sharp \in S(a)$ such that \[ p_t(v^\sharp) \leq p_{\inf}(a) + \epsilon. \] Consider a type-$v'$ agent who reports $v$ and drops out at round $t$ after the adjustment. His utility from this deviation is: \[ U_{1:t}^{\text{new}}(v; v') = U_{1:t-1}(v; v') + \gamma_t [a_t(v) \cdot v' - p_{\inf}(a_t(v))]. \] Now, in the \emph{original} mechanism (before adjustment), the agent could have reported $v^\sharp$ and dropped out at round $t$, obtaining: \[ U_{1:t}^{\text{orig}}(v^\sharp; v') = U_{1:t-1}(v^\sharp; v') + \gamma_t [a_t(v^\sharp) \cdot v' - p_t(v^\sharp)]. \] Since $v^\sharp \in S(a)$, we have $a_t(v^\sharp) = a_t(v) = a$, and since $v, v^\sharp$ both reach the same node (i.e., have the same interaction history up to round $t$), we have $U_{1:t-1}(v; v') = U_{1:t-1}(v^\sharp; v')$. Therefore: \[ U_{1:t}^{\text{new}}(v; v') = U_{1:t}^{\text{orig}}(v^\sharp; v') + \gamma_t [p_t(v^\sharp) - p_{\inf}(a)] \leq U_{1:t}^{\text{orig}}(v^\sharp; v') + \gamma_t \epsilon. \] By IC and PIR of the original mechanism: \[ U_{1:T}(v'; v') \geq U_{1:t}^{\text{orig}}(v^\sharp; v'). \] Combining: \[ U_{1:T}(v'; v') \geq U_{1:t}^{\text{new}}(v; v') - \gamma_t \epsilon. \] Since this holds for all $\epsilon > 0$, taking $\epsilon \to 0$ yields: \[ U_{1:T}(v'; v') \geq U_{1:t}^{\text{new}}(v; v'), \] as desired.

\item If $\tau > t$. By our adjustment step, the utility for an agent remains unchanged after proceeding beyond round $t+1$, since the adjustment had accounted for discounting. Hence any deviation still would not benefit the agent. 
\end{itemize}

\textbf{Part (iii). No adjustment is needed in the final period.} Fix a node at depth $T$ and two reports $v, v'$ in its candidate set with $a_T(v) = a_T(v') = a$. Since $v$ and $v'$ reach the same node, after the equalization steps in rounds $1, \ldots, T-1$ we have $a_s(v) = a_s(v')$ and $p_s(v) = p_s(v')$ for all $s < T$. By Parts (i) and (ii), the mechanism obtained after these steps is still IC, so $U_{1:T}(v; v) \ge U_{1:T}(v'; v)$, which reduces to $p_T(v) \le p_T(v')$. Exchanging the roles of $v$ and $v'$ gives $p_T(v) = p_T(v')$. Hence all reports in $S(a)$ already pay a common price in round $T$, this price equals $p_{\min}(a)$ for $a \neq 0$, and no deferral is needed. 

\textbf{Part (iv). Best response and revenue.} 
After all adjustments, all reports in the candidate set of a node at depth $t$ share the same allocations and payments in rounds $1, \ldots, t-1$, and the menu at that node offers exactly the contract $(a_t(\hat v), p_t(\hat v))$ of each report $\hat v$ in the candidate set. Hence any strategy of a type-$v$ buyer in $\cI$ traces a path in the tree and yields utility $U_{1:T}(\hat v; v)$ for some report $\hat v$ (or $U_{1:t}(\hat v; v)$ if he drops out after round $t$). By Parts (i)--(iii), the final adjusted direct mechanism is IC and PIR, so following the path $(a_t(v))_{t=1}^T$, which yields $U_{1:T}(v; v)$, is a best response for type $v$. Along this path the buyer pays the adjusted payments $p_t(v)$. Each adjustment step moves an amount $\Delta = p_t^{\text{old}}(v) - p_{\min}(a_t(v)) \ge 0$ from round $t$ to round $t+1$ and scales it by $\gamma_t/\gamma_{t+1} \ge 1$, so it changes the undiscounted revenue from report $v$ by $(\gamma_t/\gamma_{t+1} - 1)\Delta \ge 0$. Hence the revenue along this path is at least $\Rev(v; \cD)$. Finally, under the tie-breaking convention of \Cref{sec:setup}, the buyer selects, among all his best responses, one maximizing the seller's revenue, and therefore $\Rev(v; \cI) \ge \Rev(v; \cD)$. 
\end{proof}

\section{Proof of Lower Bound}
\label{sec:proof-lower-bound}

Consider the case with only two 
types $0 < \underline{v} < \overline{v}$. In general, for a truthful direct mechanism, the following principles hold:
\begin{itemize}
    \item \textbf{Lower type:} Receives an inefficient 
    allocation (the item is withheld in certain rounds), 
    but the seller need not pay information rent to the 
    buyer. 
    
    \item \textbf{Higher type:} Receives efficient 
    allocation (the item is allocated in every round), 
    but the buyer extracts information rent from the 
    seller.
\end{itemize}
The above principal generally holds for usual adverse selection problems where the agent's type is one-dimensional, and a more detailed discussion can be found in ~\cite{salanie2005economics}. 
To understand why, consider the incentive compatibility constraint: the higher-type agent $\overline{v}$ must weakly prefer reporting truthfully over misreporting as type $\underline{v}$. In a first-best scenario, the item would always be allocated and the agent always charged their true valuations. Such is not possible in the second-best scenario when the agent's type is unknown. The seller must adopt the following to ensure incentive compatibility: 

\begin{enumerate}
    \item \emph{Withhold allocation to the lower type report}: By withholding the item from type 
    $\underline{v}$ in certain rounds (setting 
    $a_t(\underline{v}) < 1$ for some $t$), the seller makes this allocation less attractive to the higher type. If the item is always allocated at the lower type, he cannot be made to pay more than his true valuation, and hence the higher type should deviate to the lower type and be sure to enjoy the surplus. 
    
    \item \emph{Leave information rent to the higher type}: 
    The seller can give the higher type a utility premium 
    beyond his participation constraint. If the higher type is left with no information rent, he would of course again deviate to reporting the lower type. The lower type would not have this issue, since deviating upward would generally net him no gain. 
\end{enumerate}
We demonstrate that the regret lower bound $\Omega(T_\gamma)$ measures a fundamental tradeoff between the two quantities. 

We state the lower bound when there are two possible types. The result of course extends to richer type spaces (e.g., $v\in [0,1]$), since the mechanism can only have more restrictions. 

\begin{theorem}\label{thm:lower-bound}
Suppose there are only two possible types $\underline{v}$ 
and $\overline{v}$. The seller's strategic regret satisfies
\[
\mathrm{Reg} \geq T_\gamma \cdot 
\frac{\underline{v}(\overline{v} - \underline{v})}{\overline{v}} 
= \Omega(T_\gamma).
\]
\end{theorem}


\szcomment{We need to modify the proof so payment can be allowed even when no transfer. I.e., fractional so that it works for menu based. Update: I think this is done. }
\begin{proof}
We break the proof into several steps. 
\paragraph{Step 1: Bounding payments from the lower type.}
Since the lower-type agent participates voluntarily 
(periodic individual rationality), we have
\[
\sum_{t=1}^T \gamma_t [a_t(\underline{v}) 
\underline{v} - p_t(\underline{v})] \geq 0.
\]
Rearranging yields
\begin{equation}\label{eq:lower-type-payment}
\sum_{t=1}^T \gamma_t p_t(\underline{v}) 
\leq \sum_{t=1}^T \gamma_t a_t(\underline{v}) \underline{v}.
\end{equation}

\paragraph{Step 2: Information rent for the higher type.}
The utility that a higher-type agent obtains by 
misreporting to $\underline{v}$ (and never dropping out) is
\begin{equation}\label{eq:deviation-utility}
\sum_{t=1}^T \gamma_t [a_t(\underline{v}) 
\overline{v} - p_t(\underline{v}) ]
\geq \sum_{t=1}^T \gamma_t [ a_t(\underline{v}) 
\overline{v} - a_t(\underv)\underline{v} ],
\end{equation}
where the inequality follows 
from~\eqref{eq:lower-type-payment}.

By incentive compatibility, the principal must provide 
the higher type with information rent at least equal to 
this deviation utility. Define
\[
\mathrm{Reg}(\overline{v}) := \sum_{t=1}^T \gamma_t 
\cdot [ {a_t(\underline{v})}  \cdot 
(\overv - \underline{v} ) ],
\]
which lower bounds the regret due to information rent 
paid to the higher type.

\paragraph{Step 3: Allocation inefficiency for the 
lower type.}
The regret from allocating inefficiently to the lower 
type is at least
\begin{align*}
T\underline{v} - \sum_{t=1}^T p_t(\underv) &\ge 
T\underline{v} - \sum_{t=1}^T a_t(\underline{v}) \cdot 
\underline{v} \\
&= \sum_{t=1}^T (1 - a_t(\underline{v})) \cdot 
\underline{v} \\
&\ge \sum_{t=1}^T \gamma_t \cdot 
(1 - a_t(\underv) ) \cdot \underline{v}.
\end{align*}

The first line above is by \Cref{lemma:non-discounted-payment} below. Define
\[
\mathrm{Reg}(\underline{v}) := \sum_{t=1}^T \gamma_t 
\cdot (1 - a_t(\underv) ) \cdot \underline{v}.
\]
which lower bounds the allocation inefficiency for the 
lower type.

\paragraph{Step 4: Combining the bounds.}
The total regret satisfies
\[
\mathrm{Reg} \geq \max\left(\mathrm{Reg}(\underline{v}), 
\mathrm{Reg}(\overline{v})\right).
\]

Observe that
\[
\frac{1}{\underline{v}} \mathrm{Reg}(\underline{v}) + 
\frac{1}{\overline{v} - \underline{v}} 
\mathrm{Reg}(\overline{v}) = T_\gamma. 
\]

By the AM-GM inequality applied to the weighted harmonic 
mean,
\begin{align*}
\mathrm{Reg} \geq T_\gamma \cdot 
\frac{\underline{v}(\overline{v} - \underline{v})}{\overline{v}} 
= \Omega(T_\gamma).
\end{align*}
This completes the proof.
\end{proof}

In step 3 above we used the below lemma. which shows that the non-discounted payment from the buyer must be upper bounded by the non-discounted utility that the buyer receives. 
\begin{lemma}
\label{lemma:non-discounted-payment}
Fix any type $v$. We have
\[
\sum_{t=1}^T  p_t(v) \le \sum_{t=1}^T a_t(v) \cdot v 
\]
\end{lemma}
\begin{proof}
By PIR, we have a sequence of participation constraints. The inequality with the summand $t=s$ to $T$ means that the agent prefers continuing participating in the contracts rather than dropping out. 
\begin{align*}
\sum_{t=s}^T \gamma_t \cdot [a_t(v) \cdot v - p_t(v)] &\ge 0 \\
\end{align*}
Since the sequence $\gamma_t$ is decreasing, we can find a sequence of appropriate positive constants $\set{c_t} > 0$ such that for any $s \in [1,T]$,
\[
\gamma_s \sum_{t=1}^{s} c_t = 1. 
\]
Then multiplying the participation constraint in round $s$ by $c_s$ and summing the sequence of inequality together gives the desired bound. 
\end{proof}

\end{document}